\documentclass{article}
\usepackage{titlefoot, amsmath, amssymb, amsthm, tikz, subfig}
\usepackage[margin = 2cm]{geometry}

\DeclareMathOperator{\trace}{trace}
\DeclareMathOperator{\diag}{diag}

\newtheorem{theorem}{Theorem}
\newtheorem{lemma}{Lemma}
\newtheorem{corollary}{Corollary}

\title{From Sharma-Mittal to von-Neumann Entropy of a Graph}

\author{Souma Mazumdar\thanks{\texttt{souma.mazumdar@bose.res.in}}$^1$, Amrik Singh\thanks{\texttt{amriksingh@iitj.ac.in}}$^2$, Supriyo Dutta\thanks{\texttt{supriyo.dutta@bose.res.in}}$^1$, Sandeep Kumar Yadav\thanks{\texttt{sy@iitj.ac.in }}$^2$, Partha Guha\thanks{\texttt{partha@bose.res.in}}$^1$}
\date{$^1$Department of Theoretical Sciences \\ S. N. Bose National Centre for Basic Sciences\\ Block - JD, Sector - III, Salt Lake City, Kolkata - 700106. \vspace{.25cm} \\ $^2$Department of Electrical Engineering \\ Indian Institute of Technology Jodhpur\\ NH 65, Nagaur Road, Karwar, Jodhpur - 342037. \vspace{.25cm} }

\begin{document}
	\maketitle
	
	\begin{abstract}
		In this article, we introduce the Sharma-Mittal entropy of a graph, which is a generalization of the existing idea of the von-Neumann entropy. The well-known R{\'e}nyi, Thallis, and von-Neumann entropies can be expressed as limiting cases of Sharma-Mittal entropy. We have explicitly calculated them for cycle, path, and complete graphs. Also, we have proposed a number of bounds for these entropies. In addition, we have also discussed the entropy of product graphs, such as Cartesian, Kronecker, Lexicographic, Strong, and Corona products. The change in entropy can also be utilized in the analysis of growing network models (Corona graphs), useful in generating complex networks.
	\end{abstract}

	\section{Introduction}
	
		``Graph entropy" was first introduced in \cite{rashevsky1955life, trucco1956note, mowshowitz1968entropy} and applied for the problems in diverse fields to characterize the structure of graphs and to cater to the needs of an application. For instance, in mathematical chemistry, the graph entropy represents the structural information of graph-based systems \cite{bonchev1983information}, and the molecular structures \cite{bertz1981first}. This idea is utilized in finding the best possible encoding of messages, where the vertices of graphs are considered as symbols \cite{korner1973coding}. It is also a measure of the structural complexity in social networks \cite{butts2001complexity,everett1985role}. There are varieties of entropy functions on graphs  defined in a different context \cite{dehmer2011history, doi:10.1002/cplx.20379}. The relation between these entropy functions and the structural properties of graphs is a fundamental topic of research \cite{dehmer2008novel, dehmer2008structural}. The eigenvalues of graphs provide an elegant tool in this context. In this article, we propose Sharma-Mittal entropy of graphs which is a generalization of a number of well-known eigenvalue based entropies. 
		
		This work is at the interface of quantum information theory, entropy, and graph theory. The graph Laplacian quantum states are represented by a Laplacian matrix of a graph. The von-Neumann entropy is the measure of quantum information in a quantum state, which has been considered as the entropy of graph. Later it was generalized as the R{\'e}nyi and Tsallis entropy. The Sharma-Mittal entropy is  generalization of all of them. It is observed that these generalized entropies are more efficient in quantum information theory problems than the von-Neumann entropy \cite{mazumdar2018sharma}. The Von Neumann entropy of graphs was  introduced in\cite{braunstein2006some} and then analyzed in a number of works \cite{simmons2017symmetric, passerini2009quantifying, anand2009entropy, han2011characterizing, han2011learning, lyons2010identities, bai2012graph, minello2018neumann}. Three different Laplacian matrices of a graph are considered in the literature, which are the combinatorial Laplacian, signless Laplacian, and normalized Laplacian matrices. These articles primarily investigate the relation between the structure of the graph and the von-Neumann entropy. The relation between the quantum entanglement in graph Laplacian quantum states and the components of the graph is investigated in  \cite{de2016interpreting}. Random graphs are also recognized in the literature of the von-Neumann entropy of graphs \cite{du2010note}. The complexity of graphs is studied from the perspectives of thermodynamics in this direction \cite{han2012graph}.
		
		To the best of our knowledge, this article is the first introduction of the Sharma-Mittal entropy based on graph Laplacian quantum states, in the literature. We extensively study the Sharma-Mittal entropy of the graphs and establish other entropies as its limiting cases. The Laplacian eigenvalues of cycle, path and complete graphs are known. We calculate the Sharma-Mittal entropy for them explicitly. Apart from calculating the entropies of the graphs, we attempt to provide a number of limiting values by calculating the upper and lower bounds for entropy functions. The product graphs, such as the corona product of graphs,  are useful in modeling complex networks. In this work, we also calculate entropies for varieties of product graphs.
		
		Our paper is organized as follows. The following section contains a very brief review of the graph Laplacian quantum states to define generalized entropy functions for graphs. Here we also calculate these entropies for a number of graphs with known spectra. We utilize combinatorial Laplacian and signless Laplacian for our investigation. In section 3, we construct a number of bounds on the entropies in terms of graph parameters. The section 4 is dedicated to the study of entropy in product graph and a network growing model. Then we conclude the article.

	\section{Graph Laplacian quantum states and generalized von-Neumann entropies}

		A graph $G=(V,E)$ consists of a set of vertices $V$, and a set $E$ of unordered pair of vertices called edges. Throughout this article, $n$ and $m$ are the numbers of vertices, and edges of a graph $G$, respectively. The adjacency matrix of a graph $A(G)=(a_{ij})_{n \times n}$
		\begin{align}
		a_{ij} = \begin{cases} 1, \; \text{if} \; (v_{i},v_{j}) \in E(G);\\ 0, \; \text{if} \;  (v_{i},v_{j}) \notin E(G). \end{cases}
		\end{align}
		The degree of a  vertex $v_{i}$ is $d_{i}=\sum\limits_{j=1}^n a_{ij}$ . The degree matrix of a graph denoted by $D(G) = \diag\{d_{1}, d_{2} \dots d_{n}\}$. The combinatorial Laplacian and signless Laplacian matrix of the graph G are defined by $L(G) = D(G) - A(G)$ and $Q(G) = D(G) + A(G)$, respectively. For simplifying the nomenclature, we drop the prefix combinatorial to mention $L(G)$. The degree of a graph is denoted by $d = \sum_{i = 1}^n d_i = 2m$. As well, $d = \trace{L(G)} = \trace{Q(G)}$, for simple graphs. The eigenvalues of Laplacian matrix are $0 \leq \lambda_1 \leq \lambda_2 \leq \dots \leq \lambda_n$. Also, $\mu_1 \leq \mu_2 \leq \dots \leq \mu_n$ are the eigenvalues of the signless Laplacian matrix. The multi-sets containing $\lambda_i$ and $\mu_i$ are called $L$-spectra and $Q$-spectra respectively. Given any real number $q$ we denote $S_{L, q}(G) = \sum_{i = 1}^n \lambda_i^q$ and $S_{Q, q}(G) = \sum_{i = 1}^n \mu_i^q$. It is shown in \cite{akbari2010relation}, that $S_{Q, q}(G) \geq S_{L, q}(G)$ for any graph $G$ and any positive real number $q$.
		
		In quantum mechanics and information theory, a quantum state is represented by a density matrix, which is positive semidefinite, Hermitian and trace-one matrix. As, $L(G)$ and $Q(G)$ are positive semidefinite and Hermitian there are quantum states represented by the density matrices $\rho_L(G)$ and $\rho_Q(G)$, where
		\begin{align}
		\rho_L(G) = \frac{1}{d}L(G) ~\text{and}~ \rho_Q(G) = \frac{1}{d}Q(G),
		\end{align}
		respectively \cite{braunstein2006laplacian,adhikari2017laplacian}. They are called graph Laplacian quantum states. The von-Neumann entropy is a measure of quantum information in a quantum states. In a symbiotic fashion, we define the von-Neumann entropy of a graph $G$ which is the von-Neumann entropy of the corresponding graph Laplacian quantum state. The von-Neumann entropies and their generalizations are functions of the eigenvalues of density matrices. Note that, if $\lambda$ is an eigenvalue of $L(G)$ then $\frac{\lambda}{d}$ is an eigenvalue of $\rho_L(G)$. Similarly, $\frac{\mu}{d}$ is an eigenvalue of $\rho_Q(G)$. The spectra of a density matrix $\rho$ is a multi-set $\Lambda(\rho)$ of eigenvalues $\gamma$ of $\rho$. Thus, $\rho = \rho_L(G)$ or $\rho_Q(G)$ as well as $\gamma = \frac{\lambda}{d}$ or $\frac{\mu}{d}$, which will be clarified from the context. Denote $S_q(G) = \sum_{i = 1}^n \gamma^q$. Therefore, $S_q(G) = \frac{1}{d^q}S_{L,q}(G)$ and $S_q(G) = \frac{1}{d^q}S_{Q,q}(G)$, when the Laplacian and signless Laplacian spectra are considered, respectively. Throughout the article, we consider logarithm with respect to the base $2$. Now, we are in position to define a number of entropies of graphs, which are as follows:

		\textbf{Sharma-Mittal entropy:} The Sharma-Mittal entropy \cite{sharma1975entropy, mittal1975some, akturk2007sharma, nielsen2011closed} of a graph $G$ is denoted by, $H_{q, r}(G)$, and defined by,
		\begin{align}\label{Sharma_Mittal_entropy}
		H_{q, r}(G) = \frac{1}{1 - r}\left[\left(S_q(G)\right)^{\frac{1 - r}{1 - q}} - 1\right],
		\end{align}
		where $q$ and $r$ are two real parameters $q > 0, q \neq 1$, and $r \neq 1$.
		
		\textbf{R{\'e}nyi entropy:} The R{\'e}nyi entropy is a limiting case of the Sharma-Mittal entropy which is defined by:
		\begin{align}\label{Renyi_entropy}
		H^{(R)}_q(G) = \lim_{r \rightarrow 1} H_{q, r}(M) = \frac {1}{1 - q}\log \left(S_q(G)\right),
		\end{align} 
		where $q \geq 0$ and $q \neq 1$.
		
		\textbf{Tsallis entropy:} The Tsallis entropy of a graph $G$ is defined by:
		\begin{align}\label{Tsallis_entropy}
		H^{(T)}_{q}(G) = \lim_{r \rightarrow q} H_{q, r}(G) = \frac{1}{1-q}\left(S_q(G) - 1\right), 
		\end{align}
		where $q \geq 0$ and $q \neq 1$.
		
		\textbf{Von-Neumann entropy:} The Sharma-Mittal entropy reduces to von-Neumann entropy when both $q \rightarrow 1$ and $r \rightarrow 1$ in equation (\ref{Sharma_Mittal_entropy}), which is
		\begin{align}\label{von-Neumann_entropy}
		H(G) = \lim_{(q, r) \rightarrow (1,1)}H_{q, r}(G) = \lim_{q \rightarrow 1}H^{(T)}_{q}(G) = -\sum_{\gamma \in \Lambda(G)} \gamma\log(\gamma).
		\end{align}
		
		The eigenvalue based network parameters are well-investigated in the literature of complex network, such as the Esterda index \cite{estrada2000characterization}. Like all of them, there are non-isomorphic graphs with equal entropies. For instance, consider the following two non-isomorphic graphs \cite{dutta2018construction}:\\
		\begin{center} 
			\begin{tikzpicture}
			\node at (-1, -.5) {$G_1 = $};
			\draw [fill] (0, 0) circle [radius = 0.07];
			\node [above] at (0, 0) {$1$};
			\draw [fill] (1, 0) circle [radius = 0.07];
			\node [above] at (1, 0) {$2$};
			\draw [fill] (2, 0) circle [radius = 0.07];
			\node [above] at (2, 0) {$3$};
			\draw [fill] (3, 0) circle [radius = 0.07];
			\node [above] at (3, 0) {$4$};
			\draw [fill] (0, -1) circle [radius = 0.07];
			\node [below left] at (0, -1) {$5$};
			\draw [fill] (1, -1) circle [radius = 0.07];
			\node [below] at (1, -1) {$6$};
			\draw [fill] (2, -1) circle [radius = 0.07];
			\node [below right] at (2, -1) {$7$};
			\draw [fill] (3, -1) circle [radius = 0.07];
			\node [below] at (3, -1) {$8$};
			\draw (0,0) --(1,0);
			\draw (1,0) --(2,0);
			\draw (0,-1) --(1,-1);
			\draw (1,-1) --(2,-1);
			\draw (1,-1) --(3,0);
			\draw (0,-1) to[out=-90,in=-90] (2,-1);
			\end{tikzpicture}
			\hspace{.5cm}
			\begin{tikzpicture}
			\node at (-1, -.5) {$G_2 = $};
			\draw [fill] (0, 0) circle [radius = 0.07];
			\node [above] at (0, 0) {$1$};
			\draw [fill] (1, 0) circle [radius = 0.07];
			\node [above] at (1, 0) {$2$};
			\draw [fill] (2, 0) circle [radius = 0.07];
			\node [above] at (2, 0) {$3$};
			\draw [fill] (3, 0) circle [radius = 0.07];
			\node [above] at (3, 0) {$4$};
			\draw [fill] (0, -1) circle [radius = 0.07];
			\node [below left] at (0, -1) {$5$};
			\draw [fill] (1, -1) circle [radius = 0.07];
			\node [below] at (1, -1) {$6$};
			\draw [fill] (2, -1) circle [radius = 0.07];
			\node [below right] at (2, -1) {$7$};
			\draw [fill] (3, -1) circle [radius = 0.07];
			\node [below] at (3, -1) {$8$};
			\draw (0,0) --(1,0);
			\draw (1,0) --(2,0);
			\draw (0,-1) --(1,-1);
			\draw (1,-1) --(2,-1);
			\draw (1,0) --(3,-1);
			\draw (0,-1) to[out=-90,in=-90] (2,-1);
			\end{tikzpicture}\\
		\end{center}
		Note that, the spectrum $\Lambda(\rho_L(G_{1})  = \Lambda(\rho_L(G_{2}) = \{0, 0, 0, \frac{1}{12}, \frac{1}{12}, \frac{1}{4}, \frac{1}{4}, \frac{1}{3}\}$. Hence, their Sharma-Mittal, R{\'e}nyi, Tsallis and von-Neumann entropies are equal, when they are calculated based on Laplacian spectra. Also, it can be easily checked that for the following two graphs \cite{dutta2018constructing} these entropies are equal when they are calculated based on their signless Laplacian spectra.\\
		\begin{center}
			\begin{tikzpicture}
			\draw [fill] (0, 0) circle [radius = 0.07];
			\node [above] at (0, 0) {$1$};
			\draw [fill] (1, 0) circle [radius = 0.07];
			\node [above] at (1, 0) {$2$};
			\draw [fill] (0, -1) circle [radius = 0.07];
			\node [below] at (0, -1) {$4$};
			\draw [fill] (1, -1) circle [radius = 0.07];
			\node [below] at (1, -1) {$3$};
			\draw (0,0) --(0,-1);
			\draw (0,0) --(1,-1);
			\draw (0,-1) --(1,-1);
			\end{tikzpicture}
			\hspace{2cm}
			\begin{tikzpicture}
			\draw [fill] (0, 0) circle [radius = 0.07];
			\node [above] at (0, 0) {$1$};
			\draw [fill] (1, 0) circle [radius = 0.07];
			\node [above] at (1, 0) {$2$};
			\draw [fill] (0, -1) circle [radius = 0.07];
			\node [below] at (0, -1) {$4$};
			\draw [fill] (1, -1) circle [radius = 0.07];
			\node [below] at (1, -1) {$3$};
			\draw (0,0) --(0,-1);
			\draw (0,-1) --(1,0);
			\draw (0,-1) --(1,-1);
			\end{tikzpicture}
		\end{center} 

		Now, we calculate the Sharma-Mittal entropy for a number of graphs having explicit expressions for their $L$-spectra. The Laplacian eigenvalues of cycle graph $C_n$ with $n$ vertices are $2-2\cos(\frac{2\pi j}{n})$ where $j = 0, 1, \dots , (n-1)$ \cite{bapat2010graphs}. Degree of a cycle graph is $2n$, Hence, the eigenvalues of $\rho_Q(C_n)$ are $\lambda = \frac{1}{2n}(2-2\cos(\frac{2\pi j}{n})), j = 0, 1, \dots n$. Now, the equation (\ref{Sharma_Mittal_entropy}) indicates that the Sharma-Mittal entropy of a cycle graph based on its Laplacian eigenvalue is 
		\begin{align}\label{sharma-mittal_entropy_for_cycle}
		\begin{split}
		H_{q,r}(C_n)  = \frac{1}{1-r}\left[\left(\frac{2}{n}\right)^{\frac{q(1-r)}{1-q}}\left(\sum\limits_{j=0}^{n-1}\sin^{2q}\left(\frac{\pi j}{n}\right)\right)^{\frac{1-r}{1-q}}-1\right].
		\end{split}
		\end{align}    
		This expression may be further simplified for integer values of $q$ \cite{da2017basic}. It also appears in the study of Weyl character formula for $\operatorname{SU}(2)$ groups. Similarly, the R{\'e}nyi, Tsallis, and von-Neumann entropy of cycle graphs are $H_{q}^{(R)}(C_n) = \frac{1}{1-q} \log\left[ \sum_{j=0}^{n-1} \left( \frac{2}{n} \right)^q \sin^{2q} \left( \frac{\pi j}{n} \right) \right]$, $H_{q}^{(T)}(C_n) = \frac{1}{1-q} \left[ \left( \frac{2}{n} \right)^q \left( \sum_{j=0}^{n-1} \sin^{2q} \left( \frac{\pi j}{n} \right) \right) - 1 \right]$, and $H(C_n) = - \sum_{j=0}^{n-1} \frac{4\sin^2 \left( \frac{\pi j}{n} \right)}{2n} \log\left( \frac{4\sin^2\left(\frac{\pi j}{n}\right)}{2n} \right)$, respectively. It can be easily verified that these entropies are the limiting cases of equation (\ref{sharma-mittal_entropy_for_cycle}).
		
		The Laplacian eigenvalues of the complete graph $K_n$ with $n$ vertices are $0$ and $n$ with multiplicity $(n-1)$. The degree of $K_n$ is $d = n(n - 1)$. Hence, the eigenvalues of the density matrix $\rho_Q(K_n)$ are $\lambda = 0$, and $\frac{1}{n-1}$ with multiplicity $(n-1)$. Therefore, the Sharma-Mittal entropy of a complete graph based on $\rho_Q$ is 
		\begin{align}
		H_{q,r}(K_n) = \frac{1}{1-r}\left[\left\{\sum\limits_{i=1}^{n-1}\left(\frac{n}{n(n-1)}\right)^q\right\}^{\frac{1-r}{1-q}}-1\right] = \frac{1}{1-r}\left[\frac{1}{(n-1)^{r-1}}-1\right].
		\end{align}
		The R{\'e}nyi, Tsallis, and von-Neumann entropy of the complete graph based on Laplacian eigenvalues are $H_{q}^{(R)}(K_n) = \frac{1}{1-q}\log\left[\frac{1}{(n-1)^{q-1}}\right]$, $H_{q}^{(T)}(K_n) = \frac{1}{1-q}\left[\frac{1}{(n-1)^{q-1}}-1\right]$, and $H(K_n) = (n-1)\log(n-1)$, respectively.
		
		The path graph $P_n$ with $n$ vertices has Laplacian eigenvalues $2 - \cos(\frac{\pi j}{n})$ where $j= 0, 1, \dots (n - 1)$. The degree for path graphs is given by $d=2(n-1)$. Thus, the eigenvalues of the density matrix $\rho_Q(P_n)$ is given by $\lambda = \frac{2-\cos(\frac{\pi j}{n})}{2(n-1)}$. Now the equation (\ref{Sharma_Mittal_entropy}) indicates that the Sharma-Mittal entropy of a $P_n$ is
		\begin{align}\label{path_graph_SM_entropy}
		\begin{split}
		H_{q,r}(P_n) & = \frac{1}{1-r}\left[\left(\sum\limits_{j=0}^{n-1}\left(\frac{2-\cos\left(\frac{\pi j}{n}\right)}{2(n-1)}\right)^q\right)^{\frac{1-r}{1-q}}-1\right]  \\
		&=\frac{1}{1-r}\left[\left(\frac{1}{2(n-1)}\right)^{\frac{q(1-r)}{1-q}}\left(\sum\limits_{j=0}^{n-1}\left(2-\cos\left(\frac{\pi j}{n}\right)\right)^q\right)^{\frac{1-r}{1-q}}-1\right].
		\end{split}
		\end{align}
		The R{\'e}nyi, Tsallis, and von-Neumann entropy of $P_n$ are given by
		\begin{equation}
		\begin{split}
		H_{q}^{(R)}(P_n)& = \frac{1}{1-q}\log\left[\left(\frac{1}{2(n-1)}\right)^q\sum\limits_{j=0}^{n-1}\left(2-\cos\left(\frac{\pi j}{n}\right)\right)^q\right], \\
		H_{q}^{(T)}(P_n) &= \frac{1}{1-q}\left[\left(\frac{1}{2(n-1)}\right)^q\sum\limits_{j=0}^{n-1}\left(2-\cos\left(\frac{\pi j}{n}\right)\right)^q-1\right], ~\text{and} \\
		H(P_n) &= -\sum\limits_{j=0}^{n-1}\frac{2-\cos\left(\frac{\pi j}{n}\right)}{2(n-1)}\log\left(\frac{2-\cos\left(\frac{\pi j}{n}\right)}{2(n-1)}\right),
		\end{split}
		\end{equation}
		respectively  as a limiting case of the equation \ref{path_graph_SM_entropy}.

		The signless Laplacian eigenvalues \cite{cvetkovic2009towards} of cycle graph $C_n$ with $n$ vertices are $2+2\cos(\frac{2\pi j}{n})$ where $j = 0, 1, \dots , (n-1)$. Degree of a cycle graph is $2n$, Hence, the eigenvalues of $\rho_Q(C_n)$ are $\lambda = \frac{1}{2n}(2+2\cos(\frac{2\pi j}{n})), j = 0, 1, \dots n$. Now the equation (\ref{Sharma_Mittal_entropy}) indicates that the Sharma-Mittal entropy of a cycle graph based on its signless Laplacian eigenvalue is 
		\begin{align}\label{sharma-mittal_entropy_for_cycle_signless}
		\begin{split}
		H_{q,r}(C_n)  = \frac{1}{1-r}\left[\left(\frac{2}{n}\right)^{\frac{q(1-r)}{1-q}}\left(\sum\limits_{j=0}^{n-1}\cos^{2q}\left(\frac{\pi j}{n}\right)\right)^{\frac{1-r}{1-q}}-1\right].
		\end{split}
		\end{align}    
		Similarly, the R{\'e}nyi, Tsallis, and von-Neumann entropy of cycle graphs are respectively given by
		\begin{align}
			\begin{split}
				H_{q}^{(R)}(C_n) & = \frac{1}{1-q}\log\left[\sum\limits_{j=0}^{n-1}\left(\frac{2}{n}\right)^q\cos^{2q}\left(\frac{\pi j}{n}\right)\right], \\
				H_{q}^{(T)}(C_n) & = \frac{1}{1-q}\left[\left(\frac{2}{n}\right)^q\left(\sum\limits_{j=0}^{n-1}\cos^{2q}\left(\frac{\pi j}{n}\right)\right)-1\right], ~\text{and} \\
				H(C_n) & = -\sum\limits_{j=0}^{n-1}\frac{4\cos^2\left(\frac{\pi j}{n}\right)}{2n}\log\left(\frac{4\cos^2\left(\frac{\pi j}{n}\right)}{2n}\right).
			\end{split}
		\end{align}
		
		The signless Laplacian eigenvalues of the complete graph $K_n$ with $n$ vertices are $2(n-1)$ and $(n-2)$ with multiplicity $(n-1)$. The degree of $K_n$ is $d = n(n - 1)$. Hence, the eigenvalues of the density matrix $\rho_Q(K_n)$ are $\lambda = \frac{2}{n}$, and $\frac{n-2}{n(n-1)}$ with multiplicity $(n-1)$. Therefore, after simplification the Sharma-Mittal entropy of a complete graph based on $\rho_Q$ is 
		\begin{equation}
		H_{q,r}(K_n) =  \frac{1}{1-r}\left[\left\{\left(\frac{n-2}{n}\right)^q(n-1)^{1-q}+\left(\frac{2}{n}\right)^q\right\}^{\frac{1-r}{1-q}}-1\right].
		\end{equation}
		The R{\'e}nyi, Tsallis, and von-Neumann entropy of the complete graph based on signless Laplacian eigenvalues are respectively given by
		\begin{align}
		\begin{split}
		H_{q}^{(R)}(K_n) &= \frac{1}{1-q}\log\left[\left\{\left(\frac{n-2}{n}\right)^q(n-1)^{1-q}+\left(\frac{2}{n}\right)^q\right\}\right], \\
		H_{q}^{(T)}(K_n) &= \frac{1}{1-q}\left[\left\{\left(\frac{n-2}{n}\right)^q(n-1)^{1-q}+\left(\frac{2}{n}\right)^q\right\}-1\right], ~\text{and} \\	
		H(K_n) &= (n-1)\log\left[\frac{n(n-1)}{n-2}\right]+\log\frac{n}{2}.
		\end{split}
		\end{align}\\
		
		The path graph $P_n$ with $n$ vertices has signless Laplacian eigenvalues $2 + 2\cos(\frac{\pi j}{n})$ where $j= 0, 1, \dots (n - 1)$. The degree for path graphs is given by $d=2(n-1)$. Thus, the eigenvalues of the density matrix $\rho_Q(P_n)$ is given by $\lambda = \frac{1+\cos(\frac{\pi j}{n})}{n-1}$. Now the equation (\ref{Sharma_Mittal_entropy}) indicates that the Sharma-Mittal entropy of a $P_n$ is
		\begin{align}
		\begin{split}
		H_{q,r}(P_n)  =  =\frac{1}{1-r}\left[\left(\frac{2}{(n-1)}\right)^{\frac{q(1-r)}{1-q}}\left(\sum\limits_{j=0}^{n-1}\cos^{2q}\left(\frac{\pi j}{2n}\right)\right)^{\frac{1-r}{1-q}}-1\right].
		\end{split}
		\end{align}
		The R{\'e}nyi, Tsallis, and von-Neumann entropy of $P_n$ are respectively given by
		\begin{align}
		\begin{split}
		H_{q}^{(R)}(P_n) &= \frac{1}{1-q}\log\left[\left(\frac{2}{n-1}\right)^q\left(\sum\limits_{j=0}^{n-1}\cos^{2q}\left(\frac{\pi j}{2n}\right)\right)\right], \\
		H_{q}^{(T)}(P_n) &= \frac{1}{1-q}\left[\left(\frac{2}{n-1}\right)^q\left(\sum\limits_{j=0}^{n-1}\cos^{2q}\left(\frac{\pi j}{2n}\right)\right)-1\right], ~\text{and} \\
		H(P_n) &= -\sum\limits_{j=0}^{n-1}\frac{2\cos^2\left(\frac{\pi j}{2n}\right)}{n-1}\log\left(\frac{2\cos^2\left(\frac{\pi j}{2n}\right)}{n-1}\right).
		\end{split}
		\end{align}
	
		In addition, we can calculate the Sharma-Mittal entropy for a number of graphs, which are not explicitly done here, such as the complete bipartite graph $K_{p,q}$. Recall that, the $L$-eigenvaluies of $K_{p,q}$ are given by $p + q, p, q$ and $0$ with multiplicity $1$, $p - 1$, $q - 1$ and $1$, respectively.

	\section{Bounds on entropies}

		There are a number of graphs whose $L$-spectra, or $Q$-spectra do not have a known general expression. For them we propose a number of bounds on their corresponding entropy functions. We first construct the bound for Sharma-Mittal entropy. Bounds of other entropies are calculated as its limiting cases. Since $S_{Q, q}(G) \geq S_{L, q}(G)$ for any graph $G$ the Sharma-Mittal, R{\'e}nyi, Tsallis and von-Neumann entropy calculated from $\rho_Q(G)$ is greater than that calculated from the $\rho_L(G)$.

		\subsection{Bounds on entropy based on $\rho_L(G)$}

			In this subsections all the entropies are calculated with respect to the eigenvalues of $\rho_L(G)$ only.
			\begin{lemma}\label{Sharma_Mittal_entropy_bound}
				Given any graph $G$
				$$H_{q,r}(G) \leq \frac{1}{1-r} \left[ \left\{ n \left( \frac{\max \left( d_{i} + d_{j} \right)}{d} \right)^q \right\}^{\frac{1 - r}{1 - q}} - 1 \right],$$
				where $d_{i}$ and $d_{j}$ are degrees of any two vertices of $G$.
			\end{lemma}			
			\begin{proof}
				For any graph the maximum Laplacian eigenvalue $\lambda_{n} \leq \max{(d_{i}+d_{j})}$ \cite{bapat2010graphs}. Therefore,
				\begin{align}
				\frac{1}{1-r}\left[\left(\sum\limits_{i=1}^n \left(\frac{\lambda_{i}}{d}\right)^q\right)^{\frac{1-r}{1-q}}-1\right] \leq \frac{1}{1-r}\left[\left(n \left(\frac{\lambda_n}{d}\right)^q\right)^{\frac{1-r}{1-q}}-1\right].
				\end{align}
				Putting $\lambda_{n} \leq \max{(d_{i}+d_{j})}$ in the above expression we get the result.
			\end{proof}
			
			The lemma suggests that an upper bound of R{\'e}nyi, Tsallis, and von-Neumann entropy are given by\\ $\frac{1}{1-q}\log\left[n\left(\frac{\max{(d_{i}+d_{j})}}{d}\right)^q\right]$, $\frac{1}{1-q}\left[n\left(\frac{\max{(d_{i}+d_{j})}}{d}\right)^q-1\right]$, and $-n\left\{\frac{\max{(d_{i}+d_{j})}}{d}\right\}\log\left\{\frac{\max{(d_{i}+d_{j})}}{d}\right\}$, respectively.			
			\begin{corollary}
				For a regular graph $G$, $H_{q,r}(G) \leq \frac{1}{1-r}\left[2^{\frac{q(1-r)}{1-q}}n^{1-r}-1\right]$.
			\end{corollary}
			\begin{proof}
				A $k$-regular graph has equal degree for all the vertices, which is $k$. Hence, $\max{(d_{i} + d{j})} = 2k$ and the degree of the graph is $nk$. Putting them in lemma \ref{Sharma_Mittal_entropy_bound} we have the result.
			\end{proof}

			Now, for a regular graph with the upper bound of R{\'e}nyi, Tsallis, and von-Neumann entropy are $\frac{1}{1-q}\log\left[2^qn^{1-q}\right]$, $\frac{1}{1-q}\left[2^qn^{1-q}-1\right]$, and $2\log\left(\frac{n}{2}\right)$, respectively.
			
			Recall that, we have denoted $S_{L,q}(G) = \sum_{i = 1}^n\lambda_i^q$, where $\lambda_i$ is an eigenvalue of $L(G)$. Given any bipartite graph it can be proved that \cite{zhou2010sum},
			\begin{align}\label{Laplacian_eigenvalue_lower_bound}
			S_{L, q}(G) \geq \left(\frac{\sum\limits_{i}d_{i}^2}{m}\right)^q+(n-2)\left(\frac{tnm}{\sum\limits_{i}d_{i}^2}\right)^{\frac{q}{n-2}},
			\end{align}
			which we use in the lemma below.
			\begin{lemma}
				The Sharma-Mittal entropy of a connected bipartite graph with $n \geq 3$ vertices, $m$ edges and $t$ spanning trees is bounded below by
				\[
				\frac{1}{1-r}\left[\frac{1}{(2m)^{\frac{q(1-r)}{1-q}}}\left[\left(\frac{\sum\limits_{i}d_{i}^2}{m}\right)^q+(n-2)\left(\frac{tnm}{\sum\limits_{i}d_{i}^2}\right)^{\frac{q}{n-2}}\right]^{\frac{1-r}{1-q}}-1\right],
				\]
				where $d_{i}$ are the degrees of the vertices
			\end{lemma}	
			
			\begin{proof}
				We know that, the degree of a graph with $m$ edges is $2m$. Therefore, the equation (\ref{Sharma_Mittal_entropy}) indicates that 
				\begin{align}
				H_{q, r}(G) = \frac{1}{1 - r}\left[\frac{1}{(2m)^{\frac{q(1-r)}{1-q}}} \left(S_{L, q}(G)\right)^{\frac{1 - r}{1 - q}} - 1\right].
				\end{align}
				Now applying equation (\ref{Laplacian_eigenvalue_lower_bound}) we find the result.
			\end{proof}

			Similarly, the lower bounds of R{\'e}nyi and Tsallis of connected bipartite graphs are respectively given by, 
			\begin{align}
				\begin{split}
					& \frac{1}{1-q}\log\left[\frac{1}{(2m)^{q}}\left[\left(\frac{\sum\limits_{i}d_{i}^2}{m}\right)^q+(n-2)\left(\frac{tnm}{\sum\limits_{i}d_{i}^2}\right)^{\frac{q}{n-2}}\right]\right], \text{and} \\
					& \frac{1}{1-q}\left[\frac{1}{(2m)^{q}}\left[\left(\frac{\sum\limits_{i}d_{i}^2}{m}\right)^q+(n-2)\left(\frac{tnm}{\sum\limits_{i}d_{i}^2}\right)^{\frac{q}{n-2}}\right]-1\right],
				\end{split}
			\end{align}
			as a limiting case of the bound on the Sharma-Mittal entropy. 	
		\subsection{Bounds on entropy based on $\rho_Q(G)$} 

			In this subsections all the entropies are calculated with respect to the eigenvalues of $\rho_Q(G)$ only.			
			\begin{lemma}
				The Sharma-Mittal entropy of a graph
				\[H_{q,r}(G) \leq \frac{1}{1-r}\left[\left\{n\left(\frac{\sqrt{4m+2(n-1)(n-2)}}{2m}\right)^q\right\}^{\frac{1-r}{1-q}}-1\right].\]
			\end{lemma}			
			\begin{proof}
				If the graph $G$ has $m$ edges then degree of the graph $d = 2m$. The maximum eigenvalue of $Q(G)$ is $\mu_n$. Therefore,
				\begin{align}
				\frac{1}{1-r}\left[\left\{\sum\limits_{i=1}^{n} \left(\frac{\mu_{i}}{d}\right)^q\right\}^{\frac{1-r}{1-q}}-1\right] \leq \frac{1}{1-r}\left[\left\{n \left(\frac{\mu_n}{2m}\right)^q\right\}^{\frac{1-r}{1-q}}-1\right].
				\end{align}
				Also, the maximum signless Laplacian eigenvalue $\mu_{n} \leq \sqrt{4m+2(n-1)(n-2)}$ \cite{cvetkovic2009towards}. Putting it in the above equation, we find the result.
			\end{proof} 

			The above lemma indicates that upper bound of R{\'e}nyi, Tsallis, and von-Neumann entropy of a graph are respectively given by  $\frac{1}{1-q}\log \left[n\left(\frac{\sqrt{4m+2(n-1)(n-2)}}{2m}\right)^q\right]$, $\frac{1}{1-q}\left[n\left(\frac{\sqrt{4m+2(n-1)(n-2)}}{2m}\right)^q-1\right]$, and\\ $-n\left(\frac{\sqrt{4m+2(n-1)(n-2)}}{2m}\right)\log\left(\frac{\sqrt{4m+2(n-1)(n-2)}}{2m}\right)$.
			
			There are bounds on the largest eigenvalue of $Q(G)$ in terms of other graph parameters. Let the maximum degree of vertices in $G$ be $\delta$, and the clique number of $G$ be $w$. It can be proved that $\mu_n \leq \delta + n(1-\frac{1}{w})$ \cite{liu2008maximum}. 
			
			\begin{lemma}
				The Sharma-Mittal entropy of a graph $G$ with maximum degree $\delta$ and clique number $w$ is bounded above by
				\[\frac{1}{1-r}\left[\left[\frac{n}{(2m)^q}\left(\delta+n\left(1-\frac{1}{w}\right)\right)^q\right]^{\frac{1-r}{1-q}}-1\right].\]
			\end{lemma}
		
			\begin{proof}
				Note that,
				\begin{align}
				S_{Q,q}(G) \leq n\left(\frac{q_{n}}{2m}\right)^q = \frac{n}{(2m)^q}\left(\delta+n\left(1-\frac{1}{w}\right)\right)^q.
				\end{align}
				Putting it in the equation (\ref{Sharma_Mittal_entropy}) we find the result. 
			\end{proof}			
			In terms of maximum degree and clique number, the upper bounds of R{\'e}nyi, Tsallis and von-Neumann entropy of a graph are $\frac{1}{1-q}\log\left[\frac{n}{(2m)^q}\left(\delta+n\left(1-\frac{1}{w}\right)\right)^q\right]$, $\frac{1}{1-q}\left[\frac{n}{(2m)^q}\left(\delta+n\left(1-\frac{1}{w}\right)\right)^q-1\right]$, and $-\frac{n}{2m}\left(\delta+n\left(1-\frac{1}{w}\right)\right)\log\left(\frac{\left(\delta+n\left(1-\frac{1}{w}\right)\right)}{2m}\right)$, respectively.
			
			A lower bound on the Sharma-Mittal entropy of $G$ can be constructed in terms of the number of its spanning subgraphs $S$. Let $nc(S)$ be the number of connected components in $S$. Then the least eigenvalue of $Q(G)$ \cite{de2011smallest}
			\begin{align}
			\mu_{1} \geq \left(\frac{n-1}{2m}\right)^{n-1}\sum\limits_{S}4^{nc(S)},
			\end{align} 
			where the sum is taken over all spanning subgraphs $S$ of $G$. It provides the following bound on the Sharma-Mittal entropy of a graph.
			\begin{lemma}
				The Sharma-Mittal entropy of a graph $G$ is bounded below by 
				\[\frac{1}{1-r}\left[\left[\frac{n}{(2m)^q}\left[\left(\frac{n-1}{2m}\right)^{n-1}\sum\limits_{S}4^{nc(S)}\right]^q\right]^{\frac{1-r}{1-q}}-1\right].\]
			\end{lemma}	
			
			The lower bounds of R{\'e}nyi, Tsallis and von-Neumann entropy are $\frac{1}{1-q}\log\left[\frac{n}{(2m)^q}\left[\left(\frac{n-1}{2m}\right)^{n-1}\sum\limits_{S}4^{nc(S)}\right]^q\right]$, \\ $\frac{1}{1-q}\left[\frac{n}{(2m)^q}\left[\left(\frac{n-1}{2m}\right)^{n-1}\sum\limits_{S}4^{nc(S)}\right]^q-1\right]$, and $-\frac{n}{2m}\left[\left(\frac{n-1}{2m}\right)^{n-1}\sum\limits_{S}4^{nc(S)}\right]\log\left[\frac{1}{2m}\left[\left(\frac{n-1}{2m}\right)^{n-1}\sum\limits_{S}4^{nc(S)}\right]\right]$, respectively.

	\section{Entropy of Product Graphs}
	
		Network modelling using product graphs \cite{hammack2011handbook} is an interesting method for generating complex networks which may capture the properties of real world networks. Consider two graphs $G_1$ and $G_2$ with $n_1$ and $n_2$ vertices, respectively. The $L$-eigenvalues of $G_1$ and $G_2$ are  $0 = \lambda^{(1)}_1 \leq \dots \leq \lambda^{(1)}_{n_1}$ and $0 \leq \lambda^{(2)}_1 \leq \dots \leq \lambda^{(2)}_{n_2}$, respectively. Similarly the signless-Laplacian eigenvalues of $G_1$ are $G_2$ are given by $0 = \mu^{(1)}_1 \leq \dots \leq \mu^{(1)}_{n_1}$ and $0 \leq \mu^{(2)}_1 \leq \dots \leq \mu^{(2)}_{n_2}$, respectively. The produch graph of $G_1$ and $G_2$ is denoted by $G$. A number of graph products has been studies in literature, for instance, the Kronecker product \cite{weichsel1962kronecker}, the corona product \cite{barik2015laplacian}, and many others. We calculate Sharma-Mittal entropy for some of them. 
			
		\textbf{Cartesian product:} The Cartesian product $G = G_1 \times G_2$ of two graphs $G_1$ and $G_2$ is a graph with vertex set $V(G_1) \times V(G_2)$ where two vertices $(u_i,v_j)$ and $(u_r,v_s)$ are adjacent in $G$ if either $u_i = u_r$ and $v_j \sim v_s$  or $u_i \sim u_r$ and $v_j = v_s$, where $\sim$ represents the adjacency relation of the respective vertices \cite{barik2015laplacian}.	
	
		\begin{lemma}
			The Sharma-Mittal entropy of the Cartesian product $G = G_1 \times G_2$ is given by,
			\begin{equation*}
			H_{q,r}(G) = \frac{1}{1-r}\left[ \left(\frac{1}{d}\right)^\frac{q(1-r)}{1-q} \left( \sum_{i=1}^{n_1} \sum_{j=1}^{n_2}(\lambda^{(1)}_i + \lambda^{(2)}_j )^q \right)^\frac{1-r}{1-q} - 1\right].
			\end{equation*}
		\end{lemma}
		\begin{proof}
			The eigenvalues of $L(G)$ are $(\lambda^{(1)}_i + \lambda^{(2)}_j )$ for $1 \leq i \leq n_1,1\leq j \leq n_2$ \cite{barik2015laplacian}. If $d$ is the degree of $G$ then Shrama-Mittal entropy is given by 
			\begin{equation}\label{eq:cartesian}
				\begin{split}
					H_{q,r}(G) & = \frac{1}{1-r}\left[  \left( \sum_{i=1}^{n_1} \sum_{j=1}^{n_2} \frac{(\lambda^{(1)}_i + \lambda^{(2)}_j )}{d}^q \right)^\frac{1-r}{1-q} - 1\right] \\
					& = \frac{1}{1-r}\left[ \left(\frac{1}{d}\right)^\frac{q(1-r)}{1-q} \left( \sum_{i=1}^{n_1} \sum_{j=1}^{n_2}(\lambda^{(1)}_i + \lambda^{(2)}_j )^q \right)^\frac{1-r}{1-q} - 1\right].
				\end{split}
			\end{equation}
		\end{proof}	
			
		As the limiting cases of the Sharma-Mittal entropy the R{\'e}nyi, Tsallis, and von-Neumann entropy of the product graph $G = G_1 \times G_2$ are respectively given by
		\begin{align}
			\begin{split}
				& H_q^{(R)}(G) = \frac{1}{1-q}\log\left[ \left(\frac{1}{d}\right)^q \left( \sum_{i=1}^{n_1} \sum_{j=1}^{n_2}(\lambda^{(1)}_i + \lambda^{(2)}_j )^q \right)\right], \\
				& H_q^{(T)}(G) = \frac{1}{1-q}\left[ \left(\frac{1}{d}\right)^q \left( \sum_{i=1}^{n_1} \sum_{j=1}^{n_2}(\lambda^{(1)}_i + \lambda^{(2)}_j )^q \right)-1\right], ~\text{and} \\
				& H(G) = -\left[  \left( \sum_{i=1}^{n_1} \sum_{j=1}^{n_2}\frac{(\lambda^{(1)}_i + \lambda^{(2)}_j )}{d}\log\left(\frac{(\lambda^{(1)}_i + \lambda^{(2)}_j )}{d}\right) \right)\right].
			\end{split}
		\end{align}		
		\textbf{Kronecker product:} The Kronecker product $G = G_1 \otimes G_2$ of two graphs $G_1$ and $G_2$ is a graph with vertex set $V(G_1) \times V(G_2)$ where two vertices $(u_i,v_j)$ and $(u_r,v_s)$ are adjacent in $G_1 \otimes G_2$ if either $u_i \sim u_r$ and $v_j \sim v_s$ \cite{barik2015laplacian}.
	
		\begin{lemma}
			Let $G_1$ and $G_2$ be connected regular graphs with regularities $K$ and $S$, respectively. Then the Sharma-Mittal entropy of $G = G_1 \otimes G_2$ is given by,
			$$H_{q,r}(G) = \frac{1}{1-r} \left[  \left( \frac{1}{d} \right) ^\frac{q(1-r}{1-q}
			\left(\sum_{i=1}^{n_1}\sum_{j=1}^{n_2}\left(K\lambda^{(2)}_j + \lambda^{(1)}_i S - \lambda^{(1)}_i\lambda^{(2)}_j\right)^q
			\right)^ \frac{1-r}{1-q}
			- 1\right].$$		
		\end{lemma}
		\begin{proof}
			The eigenvalues of $L(G)$ are $K\lambda^{(2)}_j + \lambda^{(1)}_i S - \lambda^{(1)}_i\lambda^{(2)}_j$ for $i=1,\dots,n_1$ and $j=1,\dots,n_2$ \cite{barik2015laplacian}. If $d$ is the degree of $G$ then Sharma-Mittal entropy is given by
			\begin{align}\label{eq:kronecker}
				\begin{split}
					H_{q,r}(G) &= \frac{1}{1-r} \left[	\left(	\sum_{i=1}^{n_1}\sum_{j=1}^{n_2}\left(\frac{K\lambda^{(2)}_j + \lambda^{(1)}_i S - \lambda^{(1)}_i\lambda^{(2)}_j}{d}\right)^q \right)^ \frac{1-r}{1-q} - 1\right] \\
					&= \frac{1}{1-r} \left[  \left( \frac{1}{d} \right) ^\frac{q(1-r}{1-q}
					\left(\sum_{i=1}^{n_1}\sum_{j=1}^{n_2}\left(K\lambda^{(2)}_j + \lambda^{(1)}_i S - \lambda^{(1)}_i\lambda^{(2)}_j\right)^q
					\right)^\frac{1-r}{1-q}- 1\right].
				\end{split}
			\end{align}
		\end{proof}
		
		The R{\'e}nyi entropy, Tsallis entropy and von-Neumann entropy of $G=G_1\otimes G_2$ are 
		\begin{align}
			\begin{split}
				& H_q^{(R)}(G) = \frac{1}{1-q}\log \left[  \left( \frac{1}{d} \right) ^q
				\left(\sum_{i=1}^{n_1}\sum_{j=1}^{n_2}\left(K\lambda^{(2)}_j + \lambda^{(1)}_i S - \lambda^{(1)}_i\lambda^{(2)}_j\right)^q
				\right)\right],\\				
				& H_q^{(T)}(G) = \frac{1}{1-q} \left[  \left( \frac{1}{d} \right) ^q
				\left(\sum_{i=1}^{n_1}\sum_{j=1}^{n_2}\left(K\lambda^{(2)}_j + \lambda^{(1)}_i S - \lambda^{(1)}_i\lambda^{(2)}_j\right)^q
				\right)-1\right] \text{and} \\				
				& H(G) = -\log \left[  
				\left(\sum_{i=1}^{n_1}\sum_{j=1}^{n_2} \left(\frac{K\lambda^{(2)}_j + \lambda^{(1)}_i S - \lambda^{(1)}_i\lambda^{(2)}_j}{d}\right) \log\left(\frac{K\lambda^{(2)}_j + \lambda^{(1)}_i S - \lambda^{(1)}_i\lambda^{(2)}_j}{d}\right)
				\right)\right].
			\end{split}
		\end{align}		
		
		\textbf{Strong Product:}
			The Strong Product $G = G_1 \boxtimes G_2$ of two graphs $G_1$ and $G_2$ is a graph with vertex set $V(G_1) \times V(G_2)$ where two vertices $(u_i,v_j)$ and $(u_r,v_s)$ are adjacent if either $u_i = u_r$ and $v_j \sim v_s$ in $G_2$; or $u_i \sim u_r$ in $G_1$ and $v_j=v_s$; or $u_i \sim u_r$ in $G_1$ and $v_j \sim v_s$ in $G_2$ \cite{barik2015laplacian}.
		\begin{lemma}
			Let $G_1$ and $G_2$ be connected regular graphs with regularities $K$ and $S$, respectively. Then the Sharma-Mittal entropy $G= G_1\boxtimes G_2$ is given by,
			$${H_{r,q}}(G) = {1 \over {1 - r}}\left[ {{{\left( {{1 \over d}} \right)}^{q{{1 - r} \over {1 - q}}}}{{\left( {\sum\limits_{1 = 1}^{n_1} {\sum\limits_{j = 1}^{n_2} {{{\left( {(1 + S){\lambda^{(1)} _i} + (1 + K){\lambda^{(2)}_j} - {\lambda^{(1)}_i}{\lambda^{(2)}_j}} \right)}^q}} } } \right)}^{{{1 - r} \over {1 - q}}}} - 1} \right].$$
		\end{lemma}
		\begin{proof}
			The eigenvalues of $L(G)$ are $(1 + S)\lambda^{(1)}_i + (1 + K)\lambda^{(2)}_j - \lambda^{(1)}_i+\lambda^{(2)}_j $, for $i = 1, \dots, n_1$ and $j = 1,\dots,n_2$ \cite{barik2015laplacian}. If $d$ is the degree of $G$ then the Sharma-Mittal entropy is given by	
			\begin{align}\label{eq:strong}
				\begin{split}
					{H_{r,q}}(G) & = {1 \over {1 - r}}\left[ {{{\left( {\sum\limits_{1 = 1}^{n_1} {\sum\limits_{j = 1}^{n_2} {{{\left( {({{1 + S){\lambda^{(1)} _i} + (1 + K){\lambda^{(2)} _j} - {\lambda^{(1)} _i}{\lambda^{(2)}_j}} \over d}} \right)}^q}} } } \right)}^{{{1 - r} \over {1 - q}}}} - 1} \right] \\
					& = {1 \over {1 - r}}\left[ {{{\left( {\sum\limits_{1 = 1}^{n_1} {\sum\limits_{j = 1}^{n_2} {{{\left( {({{1 + S){\lambda^{(1)} _i} + (1 + K){\lambda^{(2)} _j} - {\lambda^{(1)} _i}{\lambda^{(2)} _j}} \over d}} \right)}^q}} } } \right)}^{{{1 - r} \over {1 - q}}}} - 1} \right].
				\end{split}
			\end{align}
		\end{proof}

		As the limiting cases of the equation \ref{eq:strong}, the R{\'e}nyi, Tsallis, and von-Neumann entropy of $G=G_1\times G_2$ are
		\begin{equation}
		`	\begin{split}
				& H_q^{(R)}(G) = {1 \over {1 - q}}\log \left[ {{{\left( {{1 \over d}} \right)}^q}\left( {\sum\limits_{1 = 1}^{n_1} {\sum\limits_{j = 1}^{n_2} {{{\left( {(1 + S){\lambda^{(1)} _i} + (1 + K){\lambda^{(2)} _j} - {\lambda^{(1)} _i}{\lambda^{(2)} _j}} \right)}^q}} } } \right)} \right], \\
				& H_q^{(T)}(G) = {1 \over {1 - q}}\left[ {{{\left( {{1 \over d}} \right)}^q}\sum\limits_{1 = 1}^{n_1} {\sum\limits_{j = 1}^{n_2} {{{\left( {(1 + S){\lambda^{(1)}_i} + (1 + K){\lambda^{(2)}_j} - {\lambda^{(1)}_i}{\lambda^{(2)}_j}} \right)}^q}} }  - 1} \right], ~\text{and}	\\
				& {H}(G) = \sum\limits_{1 = 1}^{n_1} \sum\limits_{j = 1}^{n_2} \left( {{{(1 + S){\lambda^{(1)}_i} + (1 + K){\lambda^{(2)}_j} - {\lambda^{(1)}_i}{\lambda^{(2)}_j}} \over d}} \right) \log \left( {{ d \over {(1 + S){\lambda^{(1)}_i} + (1 + K){\lambda^{(2)}_j} - {\lambda^{(1)}_i}{\lambda^{(2)}_j}} }} \right)  .
			\end{split} 
		\end{equation} 
			
	
		\textbf{Lexicographic Product:} The Lexicographic Product $G = G_1\bullet G_2$ of two graphs $G_1$ and $G_2$ is a graph with vertex set $V(G_1) \times V(G_2)$ where two vertices $(u_i,v_j)$ and $(u_r,v_s)$ are adjacent if either $u_i \sim u_r$ in $G_1$ or $u_i = u_r$ and $ v_j \sim v_s$ in $G_2$ \cite{barik2015laplacian}.
		\begin{theorem}
			Let $G_1$ be a connected graph with vertices $u_1,\dots,u_{n_1}$ where degree of vertex $u_i$ is $d(u_i)$. Let $G = G_1\bullet G_2$ be the Lexicographic Product of $G_1$ are any graph $G_2$. Then the Sharma-Mittal entropy of $G$ is given by
			
			\[{H_{r,q}}(G) = {1 \over {1 - r}}\left[ {{{\left( {{1 \over d}} \right)}^{{{q(1 - r)} \over {1 - q}}}}{{\left( {\sum\limits_{i = 1}^{n_1} {{{\left( {{\lambda^{(1)}_i}{n_2}} \right)}^q} + \sum\limits_{j = 2}^{n_2} {\sum\limits_{i = 1}^{n_1} {{{\left( {{\lambda^{(2)}_j} + d({u_i}){n_2}} \right)}^q}} } } } \right)}^{{{1 - r} \over {1 - q}}}} - 1} \right].\]
		\end{theorem}
		\begin{proof}
			The eigenvalues of $L(G)$ are given by \cite{barik2015laplacian}
			\begin{enumerate}
				\item[(a)] $\lambda^{(1)}_i n_2$, for $i=1,\dots,n_1$; and
				\item[(b)] $\lambda^{(2)}_j + d(u_i){n_2}$ for $i=1,2,\dots,n_1$ and $j=2,3,\dots,n_2$.
			\end{enumerate}
			If $d$ is the degree of $G$ then the Sharma-Mittal entropy is given by,
			\begin{align}\label{eq:lexico}
				\begin{split} 
					{H_{r,q}}(G) & = {1 \over {1 - r}}\left[ {{{\left( {\sum\limits_{i = 1}^{n_1} {{{\left( {{{{\lambda^{(1)} _i}{n_2}} \over d}} \right)}^q} + \sum\limits_{j = 2}^{n_2} {\sum\limits_{i = 1}^{n_1} {{{\left( {{{{\lambda^{(2)}_j} + d({u_i}){n_2}} \over d}} \right)}^q}} } } } \right)}^{{{1 - r} \over {1 - q}}}} - 1} \right] \\
					& = {1 \over {1 - r}}\left[ {{{\left( {{1 \over d}} \right)}^{{{q(1 - r)} \over {1 - q}}}}{{\left( {\sum\limits_{i = 1}^{n_1} {{{\left( {{\lambda^{(1)}_i}{n_2}} \right)}^q} + \sum\limits_{j = 2}^{n_2} {\sum\limits_{i = 1}^{n_1} {{{\left( {{\lambda^{(2)}_j} + d({u_i}){n_2}} \right)}^q}} } } } \right)}^{{{1 - r} \over {1 - q}}}} - 1} \right].
				\end{split} 
			\end{align}
		\end{proof}
	
		The R{\'e}nyi, Tsallis and von-Neumann entropy of $G = G_1\bullet G_2$ are
		\begin{align}
			\begin{split}
				H_q^{(R)}(G) &= {1 \over {1 - q}}\log \left[ {{{\left( {{1 \over d}} \right)}^q}\left( {\sum\limits_{i = 1}^{n_1} {{{\left( {{\lambda^{(1)}_i}{n_2}} \right)}^q} + \sum\limits_{j = 2}^{n_2} {\sum\limits_{i = 1}^{n_1} {{{\left( {{\lambda^{(2)}_j} + d({u_i})n_2} \right)}^q}} } } } \right)} \right], \\
				H_q^{(T)}(G) &= {1 \over {1 - q}}\left[ {{{\left( {{1 \over d}} \right)}^q}\left( {\sum\limits_{i = 1}^{n_1} {{{\left( {{\lambda^{(1)}_i}n_2} \right)}^q} + \sum\limits_{j = 2}^{n_2} {\sum\limits_{i = 1}^{n_1} {{{\left( {{\lambda^{(2)}_j} + d({u_i})n_2} \right)}^q}} } } } \right) - 1} \right], ~\text{and} \\
				H(G) &=  - \left[ {\sum\limits_{i = 1}^{n_1} {\left( {{\lambda^{(1)}_i}n_2} \right)} \log \left( {{\lambda^{(1)}_i}n_2} \right) + \sum\limits_{j = 2}^{n_2} {\sum\limits_{i = 1}^{n_1} {\left( {{\lambda^{(2)}_j} + d({u_i})n_2} \right)\log \left( {{\lambda^{(2)}_j} + d({u_i})n_2} \right)} } } \right].\\
			\end{split}
		\end{align}
		
		\textbf{Corona product:} The corona product of two graphs $G_1$ and $G_2$ with $n_1$ and $n_2$ number of vertices is given by $G=G_1\circ G_2$. It is obtained by taking $n_1$ copies of $G_2$ and joining $i^{th}$ copy of $G_2$ with $i^{th}$ vertex of $G_1$, for each $i \in \{ 1,2,\dots,n_1\} $. An example of the corona product of two graphs $P_2$ and $P_3$ is shown in figure\ref{CP1}. 
		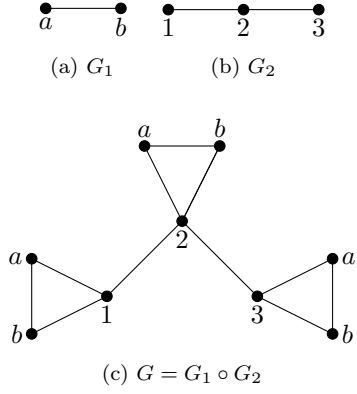
\begin{figure}[!h]						
			\centering
			\subfloat[$G_1$]{
				\centering
				\begin{tikzpicture}
				\draw [fill] (0, 0) circle [radius=0.07];
				\node [below] at (0, 0) {$a$};
				\draw [fill] (1, 0) circle [radius=0.07];
				\node [below] at (1, 0) {$b$};
				\draw (0,0) --(1,0);
				\end{tikzpicture}
			}
			\subfloat[$G_2$]{
				\centering
				\begin{tikzpicture}
				\draw [fill] (0, 0) circle [radius=0.07];
				\node [below] at (0, 0) {$1$};
				\draw [fill] (1, 0) circle [radius=0.07];
				\node [below] at (1, 0) {$2$};
				\draw [fill] (2, 0) circle [radius=0.07];
				\node [below] at (2, 0) {$3$};
				\draw (0,0) --(1,0);
				\draw (1, 0) -- (2, 0);
				\end{tikzpicture}
			}		\\
			\subfloat[$G = G_1 \circ G_2$]{
				\centering
				\begin{tikzpicture}
				\draw [fill] (0, -0.5) circle [radius = 0.07];				
				\node [left] at (0, -0.5) {$b$};
				\draw [fill] (1, 0) circle [radius = 0.07];
				\node [below] at (1, 0) {$1$};
				\draw (0,-0.5) -- (1, 0);
				\draw [fill] (0, 0.5) circle [radius = 0.07];
				\node [left] at (0, .5) {$a$};
				\draw (0,-0.5) -- (0, .5);
				\draw (0,0.5) -- (1, 0);
				\draw [fill] (2, 1) circle [radius = 0.07];
				\node [below] at (2, 1) {$2$};
				\draw (2,1) -- (1, 0);
				\draw [fill] (1.5, 2) circle [radius = 0.07];
				\node [above] at (1.5, 2) {$a$};
				\draw (1.5, 2) -- (2.5, 2);
				\draw [fill] (2.5, 2) circle [radius = 0.07];
				\node [above] at (2.5, 2) {$b$};
				\draw  (2, 1) -- (1.5, 2);
				\draw  (2, 1) -- (2.5, 2);
				\draw  (2, 1) -- (3, 0);
				\draw [fill] (3, 0) circle [radius = 0.07];
				\node [below] at (3, 0) {$3$};
				\draw  (3, 0) -- (4, -0.5);
				\draw  (3, 0) -- (4, 0.5);
				\draw  (2, 1) -- (2.5, 2);
				\draw (4, 0.5) -- (4, -0.5);
				\draw [fill] (4, -0.5) circle [radius = 0.07];
				\node [right] at (4,-0.5) {$b$};
				\draw [fill] (4, 0.5) circle [radius = 0.07];
				\node [right] at (4,0.5) {$a$};
				\end{tikzpicture}
			}
			\caption{The corona product of $G_1$ and $G_2$, $G_1\circ G_2$, after the first iteration is shown in (c) with newly added edges \label{CP1}.}
		\end{figure}
	
		\textbf{Corona Graph:} Let $G=(V,E)$ be the given graph. The corona graph, of the seed graph $G^{(0)} = G$, is give by $G^{(k + 1)} = G^{(k)} \circ G^{(0)}$, where $k > 0$ is the number of iterations of the corona product. The total number of vertices in $G^{(k)}$ are $V(1+V)^{k}$. In \cite{sharma2017structural}, it is proved that the degree distribution of $G^{(k)}$ is a power law degree distribution, which fulfills the conditions for being real world network. It is also proved that the Laplacian matrix $L(G)$, is a matrix of order $(n_1+n_1n_2)$  whose eigenvalues are given by
		\begin{enumerate}
			\item[(a)] $\frac{\lambda^{(1)}_i + n_2 + 1 \pm \sqrt{(\lambda^{(1)}_i + n_2 + 1)^2-4\lambda^{(1)}_i}}{2}$ with multiplicity $1$ for $i=1,2,\dots,n_1$, and
			\item[(b)] $\lambda^{(2)}_j+1$ with multiplicity $n_1$  for $j=2,\dots,n_2$.
		\end{enumerate}
		The result indicates the following lemma.
		\begin{lemma}
			Let $G = G_1\circ G_2$ be the Corona Product of $G_1$ and $G_2$. Then Sharma-Mittal entropy of $G$ is given by,		
			\begin{align*}
			\begin{split}		
				H_{q,r}(G) = &\frac{1}{1-r}\left[\left(\frac{1}{d}\right)^{\frac{q(1-r)}{1-q}}\left(\sum\limits_{i=1}^{n_1}\frac{\lambda^{(1)}_i + n_2 + 1 \pm \sqrt{(\lambda^{(1)}_i + n_2 +  1)^2-4\lambda^{(1)}_i}}{2}\right)^{\frac{q(1-r)}{1-q}} \right]\\
				&+ \frac{1}{1-r} \left( \left(\frac{1}{d}\right)^{\frac{q(1-r)}{1-q}} \left(\sum\limits_{j=1}^{n_2}n_1\left(\lambda^{(2)}_j+1\right)\right)^{\frac{q(1-r)}{1-q}}-1\right).		
			\end{split}
			\end{align*}
		\end{lemma}	
		The R{\'e}nyi entropy, Tsallis entropy and von-Neumann entropy of $G$ are respectively given by 
		\begin{align}
			\begin{split}			
				H_{q}^{(R)}(G) = & {1 \over {1 - q}}\left( {\log \left( {\sum\limits_{i = 1}^{n_1} {{{\left( {{{{\lambda^{(1)} _i} + n_2 + 1 \pm \sqrt {{{({\lambda^{(1)} _i} + n_2 + 1)}^2} - 4{\lambda^{(1)}_i}} } \over 2}} \right)}^q}}  + \sum\limits_{j = 1}^{n_2} {{{\left( {n_1\left( {{\lambda^{(2)}_j} + 1} \right)} \right)}^q}} } \right)}\right) \\
				& - \frac{1}{1-q} \left(q\log d \right) ,  \\	
				H_{q}^{(T)}(G) = & {1 \over {1 - q}}\left[ {{{\left( {{1 \over d}} \right)}^q}\left( {{{\left( {\sum\limits_{i = 1}^{n_1} {{{{\lambda^{(1)}_i} + n_2 + 1 \pm \sqrt {{{({\lambda^{(1)}_i} + n_2 + 1)}^2} - 4{\lambda^{(1)}_i}} } \over 2}} } \right)}^q}}\right)}\right] \\
				& + \frac{1}{1-q}\left({{\left( {{1 \over d}} \right)}^q}{{\left( {\sum\limits_{j = 1}^{n_2} {n_1\left(\lambda^{(2)}_j+1\right)} } \right)}^q}  - 1  \right)\text{and} \\			
				H(G) =	& -   \sum\limits_{i = 1}^{{n_1}} {\left( {{{\lambda _1^{(1)} + {n_2} + 1 \pm \sqrt {{{\left( {\lambda _1^{(1)} + {n_2} + 1} \right)}^2} - 4\lambda _1^{(1)}} } \over {2d}}} \right)} \\
				& \times \log \left( {{{\lambda _1^{(1)} + {n_2} + 1 \pm \sqrt {{{\left( {\lambda _1^{(1)} + {n_2} + 1} \right)}^2} - 4\lambda _1^{(1)}} } \over {2d}}} \right) - \left( {\sum\limits_{j = 1}^{{n_2}} {\left( {{{\lambda _1^{(2)} + 1} \over d}} \right)\log \left( {{{\lambda _1^{(2)} + 1} \over d}} \right)} } \right).
			\end{split}
		\end{align}
		
		Let $G_1$ be any graph, and $G_2$ be an $K$-regular graph, that the Laplacian eigenvalues \cite{sharma2017structural} of $G = G_1\circ G_2$ are given by
		\begin{enumerate}
			\item[(a)] $\frac{\lambda^{(1)}_i + n_2 + 2K+ 1 \pm \sqrt{((\lambda^{(1)}_i + n_2)-(2K + 1))^2+4m}}{2}$ with multiplicity $1$ for $i=1,2,\dots,n_1$, and
			\item[(b)] $\lambda^{(2)}_j+1$ with multiplicity $n_1$  for $j=1,\dots,n_2-1$.
		\end{enumerate}
		Now we have the following lemma.
		\begin{lemma}
			Let $G_1$ be any graph, and $G_2$ be an $K$-regular graph. Then the Sharma-Mittal entropy of $G=G_1\circ G_2$, using $\rho_L(G)$, is given by		
			\begin{equation}\label{SM_CG_REG}
			\begin{split}			
				H_{q,r}(G) = & \frac{1}{1-r}\left[\left(\frac{1}{d}\right)^{\frac{q(1-r)}{1-q}}\left(\left(\sum\limits_{i=1}^{n_1}\frac{\lambda^1_i + n_2 + 2R+ 1 \pm \sqrt{((\lambda^1_i + n_2)-(2R + 1))^2+4n_2}}{2}\right)^{\frac{q(1-r)}{1-q}}\right)\right]\\
				& + \left( \frac{1}{1-r} \left( \left(\frac{1}{d}\right)^{\frac{q(1-r)}{1-q}} \left(\sum\limits_{j=1}^{n_2-1}n\left(\lambda^2_j+1\right)\right)^{\frac{q(1-r)}{1-q}}\right)-1\right).\\
			\end{split}
			\end{equation}
		\end{lemma} 
		The R{\'e}nyi entropy, Tsallis entropy and von-Neumann entropy of $G$ are $H_{q}^{(R)}(G),H_{q}^{(T)}(G)$ and $H(G)$ respectively, where 
		\begin{align}
			\begin{split}			
				H_{q}^{(R)}(G) = & {1 \over {1 - q}}\left( {\log\left( {{{\sum\limits_{i = 1}^{n_1} {\left( {{{{\lambda^{(1)}_i} + n_2 + 2K + 1 \pm \sqrt {{{(({\lambda^{(1)}_i} + n_2) - (2K + 1))}^2} + 4n_2} } \over {2}}} \right)} }^q} + {{\sum\limits_{j = 1}^{n_2- 1} {\left( {n_1\left( {{\lambda^{(2)}_j} + 1} \right)}  \right)} }^q}} \right)}\right) \\
				&  - {1 \over {1 - q}} \left( q\log d \right), \\				
			 H_{q}^{(T)}(G) = & {1 \over {1 - q}} {{\left( {{1 \over d}} \right)}^q} \left( {{{\sum\limits_{i = 1}^{n_1} {\left( {{{{\lambda^{(1)}_i} + n_2 + 2K + 1 \pm \sqrt {{{(({\lambda^{(1)}_i} + n_2) - (2K + 1))}^2} + 4n_2} } \over 2}} \right)} }^q}} \right)\\
			 &  + {1 \over {1 - q}} \left( {{\left( {{1 \over d}} \right)}^q} \left( {{{\sum\limits_{j = 1}^{n_2 - 1} {\left( {n_1\left(\lambda^{(2)}_j+1\right)} \right)} }^q}} \right) 
			 - 1 \right) \hspace*{.2cm}  \text{and} \\	
			H(G) = & -  \sum\limits_{i = 1}^{n_1}  {{{\lambda^{(1)}_i} + n_2 + 2K + 1 \pm \sqrt {{{(({\lambda^{(1)}_i} + n_2) - (2K + 1))}^2} + 4n_2} } \over 2}\\
				&\times \log \left( {{{{\lambda^{(1)} _i} + n_2 + 2K + 1 \pm \sqrt {{{(({\lambda^{(1)}_i} + n_2) - (2K + 1))}^2} + 4n_2} } \over 2}} \right)  \\
			 &+ \left( {\sum\limits_{j = 1}^{n_2 - 1} n_1 {\left( {\left( {{\lambda^{(2)} _j} + 1} \right)\log \left( {{\lambda^{(2)}_j} + 1} \right)} \right)} } \right) .
			\end{split}
		\end{align}
		In order to define the Sharma-mittal and other entropies for laplacian spectrum of corona graphs, a function $f_L(x)$ is defined below. For an arbitrary graph and a K-regular graph on $n$ nodes, the function $f_L\colon\mathbb{R}\longrightarrow\mathbb{R}$ \cite{sharma2017structural} is given by
		\begin{align}\label{fnl}
			\begin{split}			
			{f_L}(x) &= {{x + n + 1 \pm \sqrt {{{(x + n + 1)}^2} - 4x} } \over 2},
			\end{split}
		\end{align}
		where $ \hat f_L^{j}(x) = {f_L}(\hat f_L^{j - 1}(x) + 1)$ and
		$\hat f_L^0(x) = x + 1$.

		\begin{theorem}
			Consider a Corona Graph $G^{(m)}$, $m\geq 1$ generated by a graph $G=G^{(0)}$ on $n$ vertices. Let $0\leq j \leq m$, $\varLambda(G)={\lambda_1,\dots,\lambda_n}$ and $\varGamma(G)={\nu_1,\dots,\nu_n}$ be the laplacian and signless lapalcian spectrums respectively. Then the Sharma-Mittal entropy is given by
			\[
			H_{q,r}(G)={1 \over {1 - r}}\left[ {\left( {{{(\sum\limits_{i = 2}^n {\sum\limits_{j = 0}^{m - 1} {\left( {{{\left( {n(n + 1} \right)}^{m - j - 1}}} \right)\hat f_L^j(\lambda_i)} } )}^{^{q{{1 - r} \over {1 - q}}}}} + {{(\sum\limits_{i = 1}^n {\hat f_L^m(\lambda_i)} )}^{q{{1 - r} \over {1 - q}}}}} \right) - 1} \right].
			\]
		\end{theorem}
		\begin{proof}
			If $d$ is the degree of $G$ that is sum of the degrees of all the vertices i.e. $Tr(L(G))=d$. Then the eigenvalues of the density matrix $\rho_L(G)$ \cite{sharma2017structural}, using equation \ref{fnl}, are given by 
			\begin{enumerate}
				\item[(a)] $\frac{\hat{f}_L^j(\lambda_i)}{d}\in \varLambda(G^{(m)}$, $0\leq j\leq m-1,i=2,\dots n$ with multiplicity $n(n+1)^{m-j-1}$, and
				\item[(b)] $\frac{\hat{f}_L^m(\lambda_i)}{d}\in \varLambda(G^{(m)})$, $0\leq i \leq n$ with multiplicity $1$.
			\end{enumerate}
			The Sharma-Mitatl entropy is given by,
			\begin{align}
				\begin{split}
					H_{q,r}(G)&={1 \over {1 - r}}\left[ {\left( {{{\left(\sum\limits_{i = 2}^n {\sum\limits_{j = 0}^{m - 1} {{{\left( {{{\left( {n(n + 1} \right)}^{m - j - 1}}} \right)\hat f_L^j({\lambda_i})} \over d}} } \right)}^{^{q{{1 - r} \over {1 - q}}}}} + {{\left(\sum\limits_{i = 1}^n {{{\hat f_L^m({\lambda_i})} \over d}} \right)}^{q{{1 - r} \over {1 - q}}}}} \right) - 1} \right],\\		
					&={1 \over {1 - r}}\left[ {{{\left( {{1 \over d}} \right)}^{q{{1 - r} \over {1 - q}}}}\left( {{{\left(\sum\limits_{i = 2}^n {\sum\limits_{j = 0}^{m - 1} {\left( {{{\left( {n(n + 1} \right)}^{m - j - 1}}} \right)\hat f_L^j({\lambda_i})} } \right)}^{^{q{{1 - r} \over {1 - q}}}}} + {{\left(\sum\limits_{i = 1}^n {\hat f_L^m} ({\lambda_i})\right)}^{q{{1 - r} \over {1 - q}}}}} \right) - 1} \right].
				\end{split}
			\end{align}
			If $G$ is K-regular, then $\nu_n=2K$ and
			\begin{enumerate}
				\item[(c)] $\frac{\hat{f}_L^j(\nu_i)}{d}\in \varGamma(G^{(m)})$, $0\leq j\leq m-1,i=1,\dots n-1$ with multiplicity $n(n+1)^{m-j-1}$, and
				\item[(d)] $\frac{\hat{f}_L^m(\nu_i)}{d}\in \varGamma(G^{(m)})$, $0\leq i \leq n$ with multiplicity $1$.
			\end{enumerate}
			The Sharma-Mittal entropy is given by
			\begin{align}
				\begin{split}
					H_{q,r}(G)&={1 \over {1 - r}}\left[ {\left( {{{\left(\sum\limits_{i = 1}^{n-1} {\sum\limits_{j = 0}^{m - 1} {{{\left( {{{\left( {n(n + 1)} \right)}^{m - j - 1}}} \right)\hat f_L^j({\nu _i})} \over d}} } \right)}^{^{q{{1 - r} \over {1 - q}}}}} + {{\left( \sum\limits_{i = 1}^n {{{\hat f_L^m({\nu _i})} \over d}} \right) }^{q{{1 - r} \over {1 - q}}}}} \right) - 1} \right] \\		
					&={1 \over {1 - r}}\left[ {{{\left( {{1 \over d}} \right)}^{q{{1 - r} \over {1 - q}}}}\left( {{{\left(\sum\limits_{i = 1}^{n-1} {\sum\limits_{j = 0}^{m - 1} {\left( {{{\left( {n(n + 1} \right)}^{m - j - 1}}} \right)\hat f_L^j({\nu _i})} } \right)}^{^{q{{1 - r} \over {1 - q}}}}} + {{\left(\sum\limits_{i = 1}^n {\hat f_L^m} ({\nu _i})\right)}^{q{{1 - r} \over {1 - q}}}}} \right) - 1} \right].
				\end{split}
			\end{align}		
		\end{proof}
		
		The R{\'e}nyi entropy, Tsallis entropy and von-Neumann entropy of $G^{(m)}=G^{(m-1)}\circ G^{(0)}$ are $H_{q}^{(R)}(G),H_{q}^{(T)}(G)$ and $H(G)$ respectively, where
		\begin{align}
			\begin{split}
				H_{q}^{(R)}(G)& ={1 \over {1 - q}}\left( {\log \left( {\sum\limits_{i = 2}^n {\sum\limits_{j = 0}^{m - 1} {{{\left( {\left( {{{\left( {n(n + 1} \right)}^{m - j - 1}}} \right)\hat f_L^j({\nu _i})} \right)}^q}} }  + \sum\limits_{i = 1}^n {{{\left( {\hat f_L^m({\nu _i})} \right)}^q}} } \right) - q\log d} \right), \\
				H_{q}^{(T)}(G))&= {1 \over {1 - q}}\left[ {{{\left( {{1 \over d}} \right)}^q}\left( {\left( {\sum\limits_{i = 2}^n {\sum\limits_{j = 0}^{m - 1} {{{\left( {\left( {{{\left( {n(n + 1} \right)}^{m - j - 1}}} \right)\hat f_L^j({\nu _i})} \right)}^q}} } } \right) + \left( {{{\left( {\sum\limits_{i = 1}^n {\hat f_L^m({\nu _i})} } \right)}^q}} \right)} \right) - 1} \right] \text{and} \\ 
				H(G))&=- \left( {\sum\limits_{i = 2}^n {\sum\limits_{j = 0}^{m - 1} {\left( {\left( {{{\left( {n(n + 1} \right)}^{m - j - 1}}} \right){{\hat f_L^j({\nu _i})} \over d}\log {{\left( {\hat f_L^j({\nu _i})} \right)} \over d}} \right)} }  + \sum\limits_{i = 1}^n {\left( {{{\hat f_L^m({\nu _i})} \over d} \log\left( {{{\hat f_L^m({\nu _i})} \over d}} \right)} \right)} } \right).
			\end{split}
		\end{align}
		
		The entropies for corona graph with different valus of $r$ and $q$ are plotted in figure \ref{fig:entropies}. It is clearly visible from the figure that for every value of $r$ and $q$ used, The R{\'e}nyi entropy first drops to minimum and then keep on increasing with each iteration of the corona graph and Shanon entropy monotonically increasing. Tsallis entropy and Shanon entropy are constant for each iteration of the corona product.
	
		\begin{figure}[!h]
			\subfloat{\includegraphics[width=2.8in]{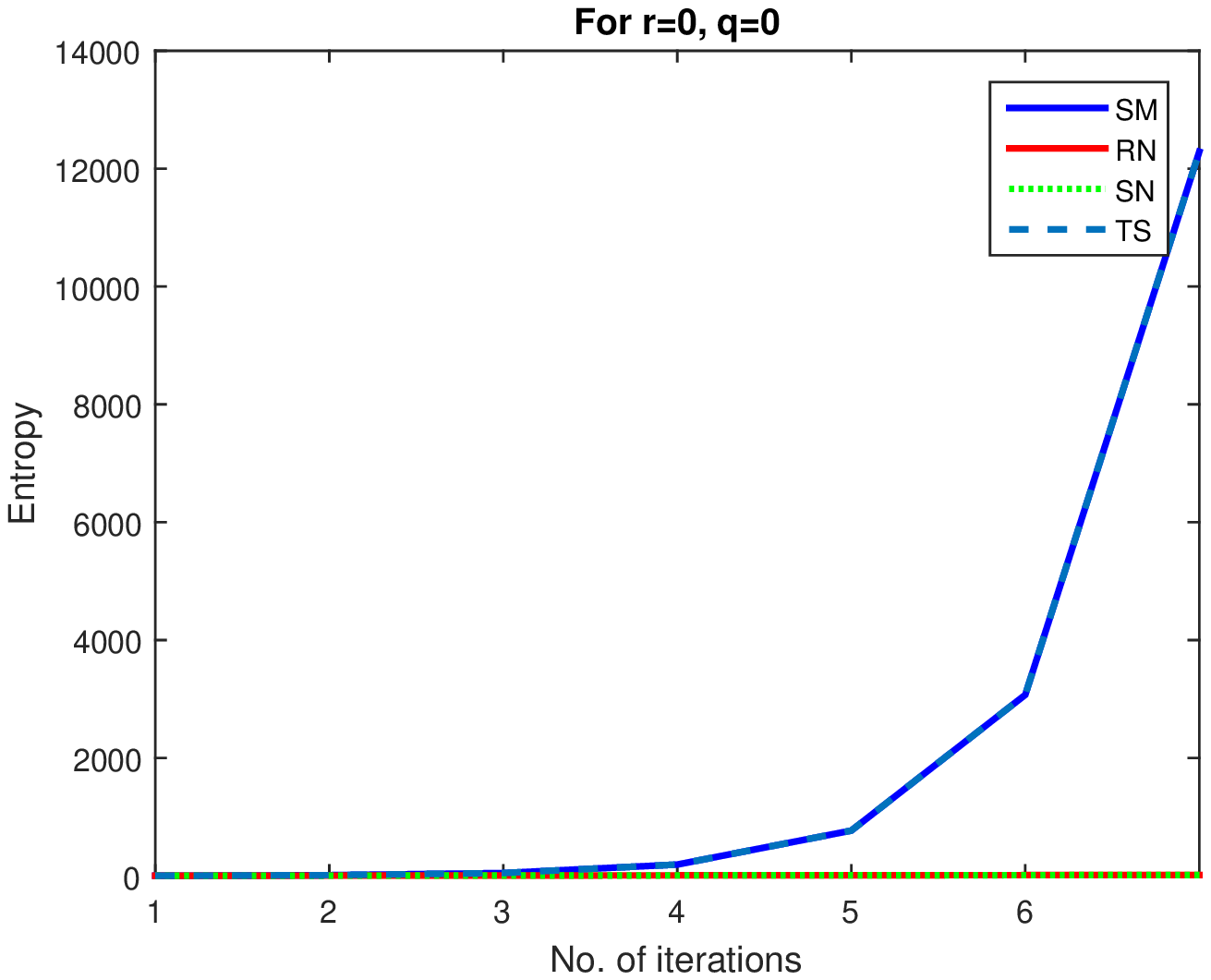}
			\label{a}}
			\hfill
			\subfloat{\includegraphics[width=2.8in]{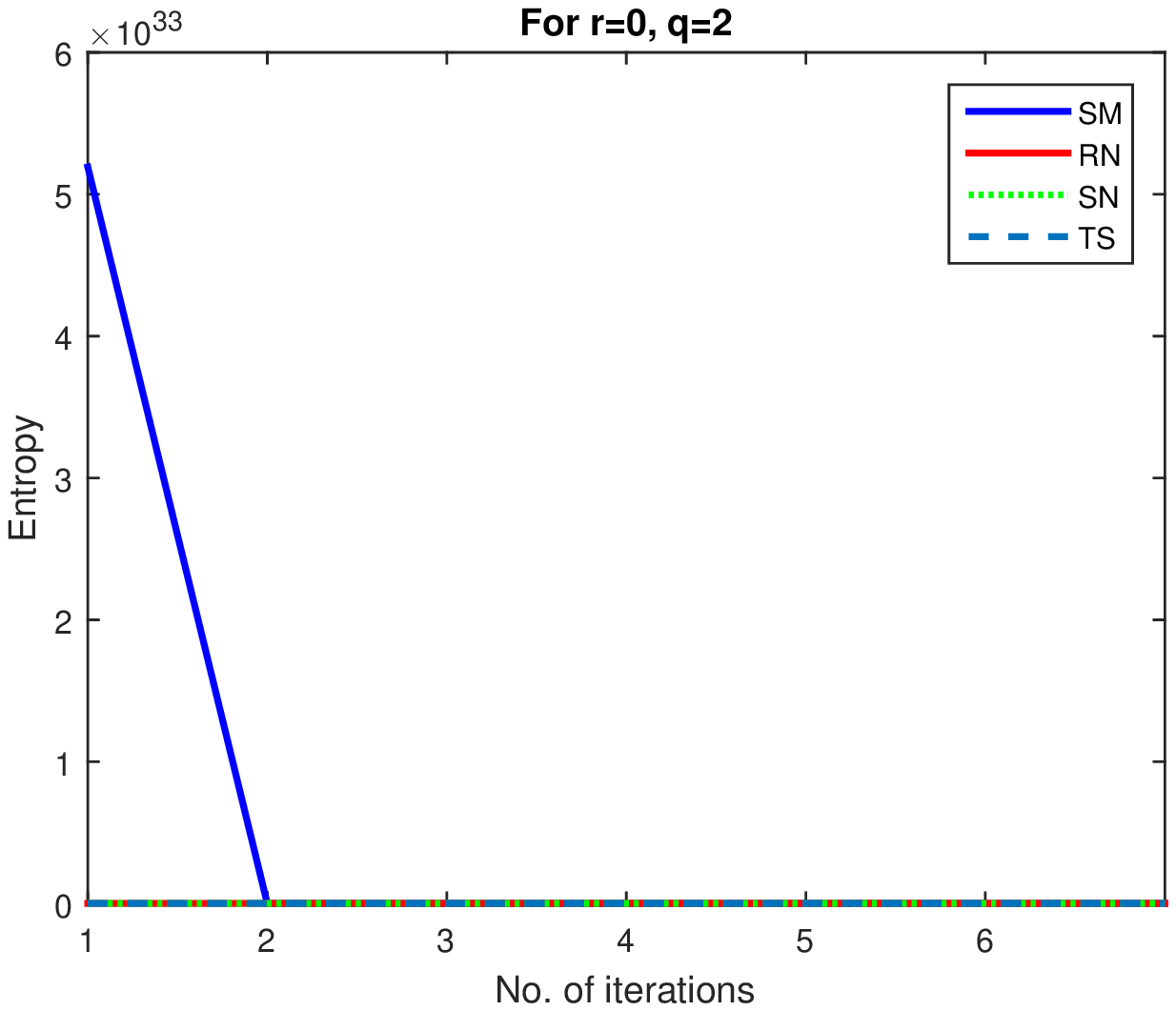}
			\label{b}}\\
			
			\subfloat{\includegraphics[width=2.8in]{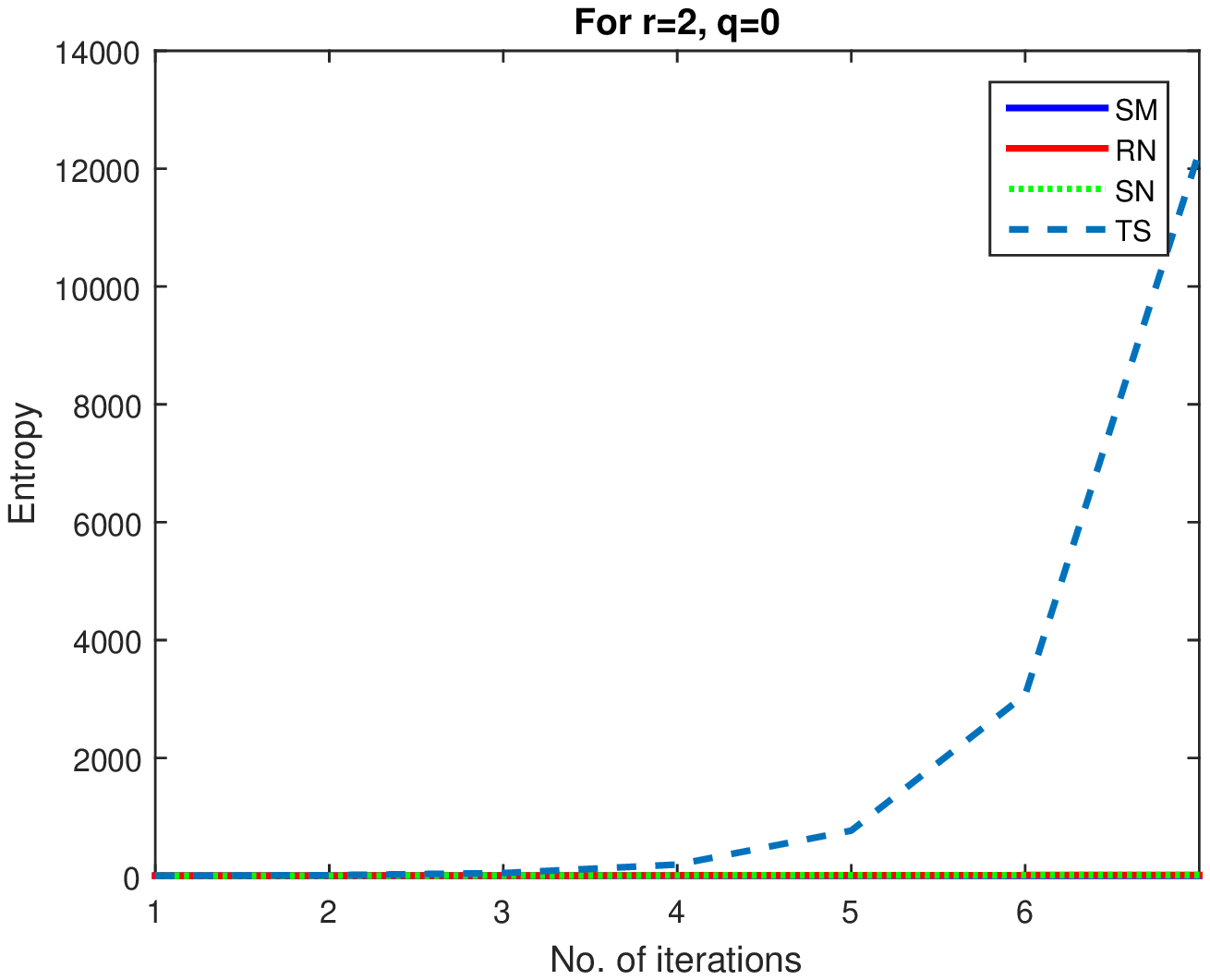}
			\label{c}}
			\hfill
			\subfloat{\includegraphics[width=2.8in]{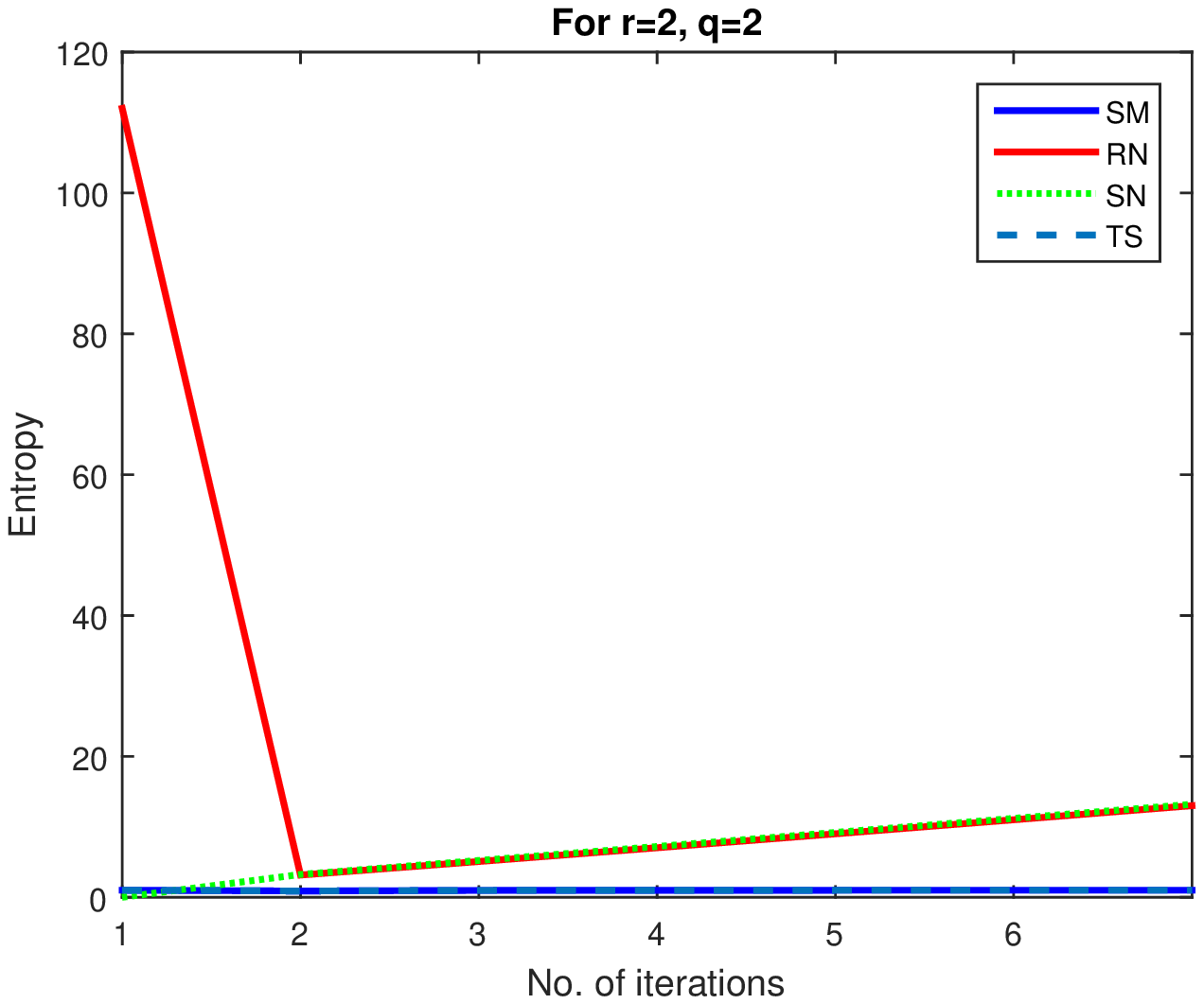}
			\label{d}}
		\caption{Entropies for six iterations ($1\leq m \leq 6$) of Corona Graph, $G^{(m)}=G^{(m-1)}\otimes G$ with different values of $r$ and $q$ where $G$ is $K_3$.}
		\label{fig:entropies}
		\end{figure}

	\section{Conclusion and open problems} 
	
		In this article we calculate the Sharma-Mittal entropy for cycle, path, and complete graph. Their R{\'e}nyi, Tsallish, and von-Neumann entropies are calculated as a limit of Sharma-Mittal entropy. A number of bound on these entropies are demonstrated. We also have studied the changes of these entropies in the formation of complex network. 
		
		It opens a number of new directions for research. These generalised entropies are essentials in the investigations of dynamical systems and thermodynamics. Thermodynamic properties of graphs in terms of Sharma-Mittal and Tsallis entropy may be investigated by interested readers.

	\section*{Acknowledgement}
	
		SM, AS, and SD have equal contribution in this work. The section 2 and 3 may be a part of SM's PhD thesis. The section 4 may be included in the PhD thesis of AS.	
		

\end{document}